\documentclass{article}
\usepackage{arxiv}
\usepackage{graphicx} 
\usepackage{epsfig} 
\usepackage{amsmath, amsfonts} 
\usepackage{amssymb}  
\usepackage{mathtools}
\usepackage{breqn}
\usepackage{float}
\usepackage{cite}
\usepackage{bm}
\usepackage{makecell}
\usepackage{subcaption}
\usepackage{caption}
\usepackage{tabularx}
\usepackage{csquotes}
\usepackage{booktabs,ragged2e}
\usepackage[flushleft]{threeparttable}

\usepackage{bigints}
\usepackage[dvipsnames]{xcolor}
\usepackage{enumitem}
\usepackage{hyperref}
\usepackage{multicol}
\usepackage{multirow}
\usepackage{threeparttable}
\usepackage{algorithm}
\usepackage{algorithmic}

\newtheorem{theorem}{Theorem}[section]

\allowdisplaybreaks

\let\boldsymbol\bm

\newcommand{\R}{\mathbb{R}}

\definecolor{codegreen}{rgb}{0,0.6,0}
\newcommand{\q}{\boldsymbol{q}}
\newcommand{\qd}{\dot{\boldsymbol{q}}}
\newcommand{\qdd}{\ddot{\boldsymbol{q}}}
\newcommand{\dq}{\boldsymbol{q}_r}

\newcommand{\ed}{\dot{\boldsymbol{e}}}

\newcommand{\norm}[1]{\lVert #1 \rVert}

\newcommand{\bs}[1]{\boldsymbol{ #1 }}
\newcommand{\ol}[1]{\overline{ #1 }}
\newcommand{\ul}[1]{\underline{ #1 }}
\newcommand{\paranthesis}[1]{\left( #1 \right)}

\begin{document}
\title{Derivative-Free Generalized Multivariable Super-Twisting Control for Constrained Euler–Lagrange Systems}

\author{Chidre Shravista Kashyap \emph{Member,~IEEE}  and  Jishnu Keshavan, \emph{Member,~IEEE}
\thanks{This research was supported in part by the MOE Grant STARS-2/2023-0265}
\thanks{The authors are with the Department of Mechanical Engineering and Robert Bosch Centre for Cyber-Physical Systems, Indian Institute of Science, Bangalore, Karnataka~560012, India (email :
chidres@iisc.ac.in, kjishnu@iisc.ac.in).}}
\maketitle

\begin{abstract}
Robotic manipulators performing payload lift-and-transfer must track prescribed trajectories under abrupt load changes, with limited actuation and inaccurate plant knowledge. Such plants are inherently described by multivariable Euler--Lagrange~(EL) dynamics with cross-coupling established through inertial and Coriolis terms. The super-twisting algorithm~(STA) is the standard
approach to regulating a first-order sliding mode~(FOSM) in such settings, yet existing designs are largely restricted to scalar sliding variables, and rest on a disturbance-derivative assumption that EL systems violate. In particular, the FOSM dynamics depend on the disturbance and become
discontinuous under the step changes in load produced by payload engagement and release, rendering that derivative ill-defined. This study resolves both with a derivative-free generalized multivariable super-twisting framework. The discontinuous sign-function integral term is replaced with a continuous fractional-power term whose dynamics are independent of the disturbance derivative, while the known state-dependent scaling of the perturbation is incorporated into both the control law and the adaptive gain update. A time-shifting barrier function guarantees prescribed-time convergence
to a user-specified bound independently of initial conditions, with the design relying only on mass-matrix norm bounds. Simulation and experimental results on a 7-DoF manipulator under actuator saturation are used to validate the proposed scheme.
\end{abstract}
\keywords{
Euler-Lagrange systems, industrial motion control, actuator saturation, prescribed-time convergence, super-twisting control} 

\section{INTRODUCTION}
\label{introduction}
Robotic manipulators executing payload lift-and-transfer operations must track prescribed
trajectories under abrupt changes in external load, with limited actuation and without
accurate knowledge of the plant. Such tasks reduce to regulating a first-order sliding
mode~(FOSM) of an uncertain Euler-Lagrange~(EL) system using a continuous and
bounded control action. EL systems are multivariable in nature, with the mass, Coriolis
and friction terms coupling all degrees of freedom, and the resulting FOSM dynamics
becomes discontinuous under the non-smooth disturbances that load acquisition and release
produce. The super-twisting algorithm~(STA) is the standard method for this
setting~\cite{Levant:1998}, with multivariable extensions
in~\cite{LopezCaamal:2019, Moreno:2022}, yet a structural incompatibility exists between
the STA and the EL system class that existing designs do not resolve. The source of this incompatibility lies in the integral term of the STA. In all existing formulations, this term contains a discontinuous sign function whose closed-loop dynamics explicitly involve the disturbance derivative. Under any non-smooth disturbance, this derivative is
ill-defined, and the standard assumption of a bounded disturbance derivative is violated. Consequently, the stability
guarantees of existing STA
designs~\cite{Levant:1998,Golkani:2018,Castillo:2016b,Seeber:2020,
Edwards_2016,Shtessel:2012,Obeid:2020,tian2019adaptive,SPSTC} cannot be certified for EL
systems under general bounded disturbances.

\begin{table*}[t]
    \centering
    \caption{Qualitative comparison of the various STA schemes considered in this study.
    ``Deriv.-free'': no assumption on the disturbance derivative is required.
    ``Presc.\ $\varepsilon$'': convergence to a user-specified bound $\varepsilon>0$ is guaranteed.
    ``Presc.\ $t_c$'': convergence time is user-specified independently of initial conditions and unknown system parameters.
    ``Feasibility'': a formal sufficient condition on control authority is derived to guarantee task feasibility. ``$\alpha$, $\beta$'': exponents of the proportional and integral terms of the STA respectively, with $\beta = 0$ denoting a discontinuous integral term.}
    \begin{tabular}{|c|c|c|c|c|c|c|c|c|}
        \hline
        Scheme & Multi- & Input & Perturbation & Deriv.- & Presc. & Presc. & Feasi- & $\alpha,\,\beta$\\
        {} & variable & sat. & bounds & free & $\varepsilon$ & $t_c$ & bility & \\
        \hline
        \hline
        \makecell{Conventional\\ STA \cite{Levant:1998}}
          & No & No & Known & No & No & No & No
          & \makecell{$\alpha=1/2,$\\$\beta=0$} \\\hline
        \makecell{Saturated\\ STA \cite{Golkani:2018,Castillo:2016b}}
          & No & Yes & Known & No & No & No & No
          & \makecell{$\alpha=1/2,$\\$\beta=0$} \\\hline
        \makecell{Conditioned \\STA \cite{Seeber:2020}}
          & No & Yes & Known & No & No & No & No
          &\makecell{ $\alpha=1/2,$\\$\beta=0$} \\\hline
        \makecell{Adaptive STA\\ \cite{Edwards_2016,Obeid:2020,tian2019adaptive,Shtessel:2012}}
          & No & No & Unknown & No & Yes & No & No
          & \makecell{$\alpha=1/2,$\\$\beta=0$} \\\hline
        \makecell{Saturation-\\tolerant STA \cite{SPSTC}}
          & No & Yes & Unknown & No & Yes & No & No
          & \makecell{$\alpha=1/2,$\\$\beta=0$} \\\hline
        \makecell{Adaptive dual-\\layer GSTA \cite{Nunes:2026}}
          & No & No & Unknown & No & No & No & No
          & \makecell{$\alpha=1/2,$\\$\beta=0$} \\\hline
        \makecell{Generalized \\STA \cite{Mei:2023}}
          & No & No & Known & No & No & No & No
          & \makecell{$\frac{1}{2}\leq\alpha<1,$\\$\beta=2\alpha{-}1$} \\\hline
        \makecell{Adaptive \\saturated STA \cite{Keshavan:2026}}
          & No & Yes & Unknown & No & Yes & No & No
          & \makecell{$\frac{1}{2}\leq\alpha<1,$\\$\beta=2\alpha{-}1$} \\
        \hline
        \hline
        \makecell{Multivariable\\ STA \cite{LopezCaamal:2019}}
          & Yes & No & Known & No & No & No & No
          & \makecell{$\alpha=1/2,$\\$\beta=0$} \\\hline
        \makecell{Multivariable\\ STA \cite{Moreno:2022}}
          & Yes & No & Known & No & No & No & No
          & \makecell{$\alpha=1/2,$\\$\beta=0$} \\\hline
       \makecell{ Adaptive \\multivariable STA \cite{Borlaug:2022}}
          & Yes & No & Unknown & No & No & No & No
          & \makecell{$\alpha=1/2,$\\$\beta=0$} \\\hline
        \hline
        \hline
        Proposed (\ref{control_policy})
          & Yes & Yes & Unknown & Yes & Yes & Yes & Yes
          & \makecell{$\frac{1}{2}<\alpha<1,$\\$\beta=2\alpha{-}1$} \\
        \hline
        \hline
    \end{tabular}
    \label{tab:qual_comparison}
\end{table*}

This limitation persists across the literature, as summarized in
Table~\ref{tab:qual_comparison}. Existing multivariable
designs~\cite{LopezCaamal:2019,Moreno:2022,Borlaug:2022} and the adaptive dual-layer
design~\cite{Nunes:2026} all retain the bounded derivative requirement. A recent
generalized STA~\cite{Mei:2023} introduces the fractional-power integral structure for scalar systems with known
bounds, and its adaptive extension~\cite{Keshavan:2026} adds saturation handling and
unknown bounds, but neither achieves prescribed-time convergence. Saturated
designs~\cite{Golkani:2018,Castillo:2016b,Seeber:2020,SPSTC} address
integral windup but remain scalar, while prescribed-time methods based on time-varying
gains~\cite{Song:2017,Song:2019} suffer from an infinite-gain singularity incompatible
with saturation. Barrier function designs~\cite{Obeid:2020,SPSTC} guarantee convergence
to a prescribed bound but cannot recover once a non-smooth disturbance drives the
trajectory outside the envelope, providing no global stability guarantee precisely in the
scenarios that motivate this work. Thus no existing design, scalar or multivariable,
simultaneously removes the derivative assumption, handles input saturation, and provides
prescribed-time convergence for multivariable EL systems. The proposed framework resolves
this by eliminating the derivative dependence at its source; the main contributions are
as follows.

\begin{enumerate}

\item \emph{Derivative-free multivariable STC for EL systems.} A continuous
fractional-power integral structure is introduced whose dynamics are independent of the
disturbance derivative, admitting all bounded disturbances --- including step changes and
discontinuous signals --- that are excluded by every existing STA
design~\cite{Levant:1998,Golkani:2018,Castillo:2016b,Seeber:2020,
Edwards_2016,Shtessel:2012,Obeid:2020,tian2019adaptive,SPSTC,LopezCaamal:2019,
Moreno:2022,Borlaug:2022,Nunes:2026}, and mitigating chattering without approximation.
The multivariable EL setting requires a Kronecker product Lyapunov structure and matrix
saturation coefficient that do not arise in prior scalar
designs~\cite{Mei:2023,Keshavan:2026}.

\item \emph{Joint adaptive and feedforward treatment of the EL perturbation structure.}
The state-dependent perturbation is handled by gain adaptation for the unknown magnitude
together with direct feedforward compensation of the known state-dependent factor. The
resulting design requires no knowledge of the Coriolis, gravity or friction terms, nor of
their bounds, and uses only norm bounds on the mass matrix, so that no identification of
the manipulator dynamics is needed before deployment. This addresses the difficulty
identified in~\cite{Moreno:2022,Roy2020,Baldi2020} for multivariable EL systems with
structured perturbations.

\item \emph{Global prescribed-time convergence via time-shifting barrier function.} A
time-shifting mechanism resets the convergence envelope whenever the trajectory escapes
the prescribed bound, ensuring global stability and recovery for all bounded
disturbances --- the class that causes prior barrier designs to
fail~\cite{Obeid:2020,SPSTC}. Convergence to $\varepsilon$ is guaranteed within $t_c$
independently of initial conditions and unknown system parameters, with sublinear scaling
of the steady-state bound --- more favorable than the linear scaling
of~\cite{Mei:2023} --- and a bounded adaptive gain that avoids the infinite-gain
singularity of~\cite{Song:2017,Song:2019}.

\item \emph{Saturation-consistent design with relaxed gain conditions.} Input saturation
is incorporated via a nonlinear equality transformation with a filtered saturation
coefficient~\cite{Shao:2022}, ensuring bounded control action and mitigating
integral windup without switching. The admissible gain conditions are strictly more
relaxed than those of~\cite{Keshavan:2026}.

\end{enumerate}

Section~\ref{section_2} presents the problem formulation and preliminaries;
Section~\ref{section_3} develops the proposed control policy and its stability analysis;
Section~\ref{section4} presents simulation and experimental results on the Franka
Research~3 manipulator; Section~\ref{section5} concludes.

\section{Preliminaries}
\label{section_2}
\subsection{Notations, Definitions and Lemmas}
\label{subsection_21}
Throughout the rest of this article, the following notations will be used. For a given vector $\boldsymbol{x} = [x_1 ,...,x_n ]^{\top}\in \mathbb{R}^n$, and for a real number $r\in\mathbb{R}$, the multivariable sign function is defined as $\lceil\boldsymbol{x}\rfloor^0=\boldsymbol{x}/||\boldsymbol{x}||$, and $\lceil \boldsymbol{x}\rfloor^r=||\boldsymbol{x}||^r \lceil\boldsymbol{x}\rfloor^0$, where $||.||$ denotes the 2-norm in $\mathbb{R}^n$. The time-derivative of the multivariable sign function is given by 
\begin{eqnarray}
    \label{sign_deriv}
    \frac{d\lceil\boldsymbol{x}\rfloor^0}{dt}=\frac{1}{||\boldsymbol{x}||}\biggl[\bs{I}-\frac{\boldsymbol{x}\boldsymbol{x}^{\top}}{||\boldsymbol{x}||^2}\biggr]\dot{\boldsymbol{x}},
\end{eqnarray}
where $\bs{I}$ is the identity matrix of dimension $n$. Finally, for a function $\boldsymbol{f}(t)\in\mathbb{R}^n\,\forall n\in[1,\infty)$, $\boldsymbol{f}(t)\in\mathcal{L}_{\infty}$ when $\sup_{t} ||\boldsymbol{f}(t)||<\infty$. 

Consider the following nonlinear dynamical system:
\begin{eqnarray}
    \label{sys_eqns_0}
    \dot{\boldsymbol{x}}=\boldsymbol{f}(\boldsymbol{x}),\,\boldsymbol{x}(0)=\boldsymbol{x}_0\in\mathbb{R}^n,
\end{eqnarray}
where $\boldsymbol{x}\in\mathbb{R}^n$ is the system state, $\boldsymbol{f}(\boldsymbol{x}):\mathbb{R}^n\rightarrow\mathbb{R}^n$ is a possibly discontinuous vector field with $\boldsymbol{f}(\boldsymbol{0})=\boldsymbol{0}$. The solutions of (\ref{sys_eqns_0}) are understood in the Filippov sense \cite{Levant:2005}.

\emph{Definition 1:} For a user-specified constant $T_c>0$ and bound  $\varepsilon>0$, system (\ref{sys_eqns_0}) is said to achieve practical prescribed-time convergence if $||\bs{x}(t)||<\varepsilon\,\forall\,t\geq T_c,\forall\,\boldsymbol{x}_0\in\mathbb{R}^n$.

\emph{Lemma 1 \cite[Theorem~4.2.12]{Horn:2012}:} If $\boldsymbol{A} \in \mathbb{R}^{p \times p}$ and $\boldsymbol{B} \in \mathbb{R}^{n \times n}$ 
are both symmetric positive definite matrices, then $\boldsymbol{A} \otimes \boldsymbol{B} 
\in \mathbb{R}^{np \times np}$ is also symmetric positive definite, with 
$\lambda_{\min}\{\boldsymbol{A} \otimes \boldsymbol{B}\} = \lambda_{\min}\{\boldsymbol{A}\}
\cdot\lambda_{\min}\{\boldsymbol{B}\}$.

\subsection{Problem Statement}
\label{subsection_22}
 Consider the class of multivariable systems represented as: 
\begin{eqnarray}
    \bs{M}(\q)\qdd + \bs{C}(\q, \qd)\qd + \bs{G}(\q) + \bs{F}(\qd) + \boldsymbol{d}(t) = \text{sat}(\bs{u}(t)),
    \label{EL_eqns}
\end{eqnarray}
where $\q:\R_{\geq0}\rightarrow\R^n$ is the generalized position coordinate, and $\qd,\ \qdd$ are the generalized velocity and acceleration coordinates respectively. Moreover, $\bs{M}(\q) \in \R^{n \times n}$ is the mass matrix, $\bs{C}(\q, \qd) \in \R^{n \times n}$ is the Coriolis matrix, $\bs{G}(\q) \in \R^n$ arises due to the gravity vector, $\bs{F}(\qd) \in \R^n$ is attributed to damping and friction forces, $\bs{d}(t) \in \R^n$ is the external disturbance, and $\bs{u} \in \R^n$ is the control input. For brevity, when a symbol’s functional dependence is clear, its arguments and brackets are omitted; e.g., $\bs{C}(\q,\qd)$ and $\bs{d}(t)$ are written as $\bs{C}$ and $\bs{d}$, respectively. 

The nonlinear saturation function $\text{sat}(\bs{u})=[\text{sat}(u_1),..,\text{sat}(u_n)]^{\top}$ is given by
\begin{align}
    \label{sat_fun_defn}
    \text{sat}(u_i)=\begin{cases}
    u_i,\,\,\,\,\,\,\,\,\,\,\,\,\,\,\,\,\text{if}\,\,\,\,\,\,|u_i|\leq \overline{u}_i\\
    \overline{u}_i\lceil u_i\rfloor^0,\,\,\,\,\,\text{if}\,\,\,\,\,\,|{u}_i|\,>\,\overline{u}_i,\,i=1,..,n,\\
    \end{cases}
\end{align}
where $\overline{{\bs{u}}}=[\overline{{u}}_1,..,\overline{{u}}_n]^{\top}>\bs{0}_n$ is the known constant input saturation bound. System (\ref{EL_eqns}) is assumed to be input-to-state stable \cite{Wen:2011}. 

These EL systems satisfy the following properties for some positive real constants $\ul{M},\ \ol{M},\ \ul{m},\ \ol{m},\ \ol{C},\ \ol{G},\ \ol{F}$ that represent bounds on the norm of the system matrices \cite[Chapter 2]{EL_props},  \cite{Robot_Control:1990}.

\emph{Property 1:} The mass matrix $\boldsymbol{M}(\boldsymbol{q})$ and its inverse $\boldsymbol{M}(\boldsymbol{q})^{-1}$ are symmetric and positive definite, which implies that 
\begin{eqnarray}
\label{prop_mass}
0{<}\ul{M}\bs{I} {\leq} \bs{M}(\bs{q}) {\leq} \ol{M}\bs{I},\,0{<}\ol{M}^{-1}\bs{I} {\leq} \bs{M}(\bs{q})^{-1} {\leq} \ul{M}^{-1}\bs{I}.
\end{eqnarray}

\emph{Property 2:} $\norm{\bs{C}}\leq \ol{C}\norm{\qd}$, $\norm{\bs{G}}\leq \ol{G}$, $\norm{\bs{F}}\leq \ol{F}\norm{\qd}$.

Properties 1-2 are naturally satisfied by most mechanical systems with Euler-Lagrangian structure (\ref{EL_eqns}). In this study, the mass matrix $\bs{M}$ is assumed to be unknown; only the bounds in (\ref{prop_mass}) are assumed known. Moreover, the rest of the matrices $\bs{C},\ \bs{G},\ \text{and}\ \bs{F}$ in (\ref{EL_eqns}), along with their bounds $\ol{C},\,\ol{G},\,\ol{F}$ are assumed to be unknown for deriving the control law. We now invoke the following assumptions:

\emph{Assumption 1:} The external disturbance $\bs{d}(t)$ is bounded by some unknown positive  constant, $\ol{d}$, i.e. $\norm{\bs{d}(t)}\leq \ol{d},\ \forall\ t\geq 0$.

\emph{Assumption 2:} For a given reference trajectory $\dq\in\mathbb{R}^n$, there exist unknown positive constants $\ol{q}_{rd},\,\ol{q}_{rdd}$ such that $\norm{\dot{\bs{q}}_r(t)} \leq \ol{q}_{rd},\,\norm{\ddot{\bs{q}}_r(t)} \leq \ol{q}_{rdd}\forall\,t\geq0$.

Let $\bs{e}(t)=\bs{q}(t)-\bs{q}_r(t)$, and define the system state as $\boldsymbol{x}(t)=[\bs{e}(t)^{\top},\,\dot{\bs{e}}(t)^{\top}]^{\top}$. We now introduce the PD sliding variable as $\bs{s}(t)=\dot{\boldsymbol{e}}(t)+\bs{\Gamma}\boldsymbol{e}(t)$, where $\bs{\Gamma}\in\mathbb{R}^{n\times n}$ is a positive constant matrix. 

We now introduce $\boldsymbol{M}_0$ as a constant symmetric positive definite matrix, and obtain the time-derivative of $\bs{s}(t)$ as,
\begin{eqnarray}
    \label{sm_dyn}
    \dot{\bs{s}}{=}\ddot{\bs{q}}-\ddot{\bs{q}}_r+\bs{\Gamma}\dot{\bs{e}}{=}
    \bs{M}_0^{-1}\text{sat}(\bs{u}(t))+\bs{\rho}(t,\boldsymbol{x}),
\end{eqnarray}
where, invoking (\ref{EL_eqns}), we have,
\begin{eqnarray}
    \label{rho_eqn}
    \boldsymbol{\rho}(t,\bs{x}){=}(\bs{I}{-}\bs{M}_0^{-1}\bs{M})\ddot{\bs{e}}{+}\bs{\Gamma}\dot{\bs{e}}{-}\bs{M}_0^{-1}\bs{C}\dot{\bs{q}}{-}\bs{M}_0^{-1}\bs{G}\nonumber\\{-}\bs{M}_0^{-1}\bs{F}{-}\bs{M}_0^{-1}\bs{d}{-}\bs{M}_0^{-1}\bs{M}\ddot{\bs{q}}_r.
\end{eqnarray}
Given the bounds on $\bs{M}(\bs{q})$ invoked in Property 1, it is always possible to design $\bs{M}_0$ as a constant diagonal matrix to satisfy the condition $||\bs{I}-\bs{M}_0^{-1}\bs{M}||<1$ \cite{Roy:2019}, thus $\bs{M}_0$ serves as a surrogate of the mass matrix in controller design. Invoking Properties 1-2, we have,

\begin{flalign}
    \norm{ \bs{\rho}(t,\bs{x})} \leq& \ol{m}\,\ol{u} {+} \paranthesis{m_0{+}\ol{m}}\paranthesis{\ol{G}{ +} \ol{d} {+} \ol{q}_{rd}\ol{F} {+} \ol{q}_{rd}^2 \ol{C}} {+} \ol{q}_{rdd} \nonumber \\
     &+ \paranthesis{\norm{\bs{\Gamma}} + (m_0+\ol{m})\paranthesis{\ol{F} + 2\ol{q}_{rd}\ol{C}}}\norm{\ed} \nonumber\\
     &+ (m_0 + \ol{m})\ol{C}\norm{\ed}^2  
     \leq\rho_0(\bs{x})\ol{\rho}. \label{rho_bound}
\end{flalign}

where $m_0=\norm{\bs{M}_0},\ol{m}{=}\ul{M}^{-1}$, $\rho_0(\bs{x}){=}\text{max}\{1,||\dot{\bs{e}}||,||\dot{\bs{e}}||^2\}$ and 
$\ol{\rho} = 3\text{max}((\ol{m}\ \ol{u} + \paranthesis{m_0+\ol{m}}\paranthesis{\ol{G} + \ol{d} + \ol{q}_{rd}\ol{F} + \ol{q}_{rd}^2 \ol{C}} + \ol{q}_{rdd}), \paranthesis{\norm{\bs{\Gamma}} + (m_0+\ol{m})\paranthesis{\ol{F} + 2\ol{q}_{rd}\ol{C}}}, (m_0 + \ol{m})\ol{C})$. Observe that $\rho_0(\boldsymbol{x})$ is an available signal for feedback satisfying $\rho_0(\bs{x})\geq 1\,\forall\bs{x}$, while $\ol{\rho}$ remains unknown. Then, given any initial state $\boldsymbol{x}(0)=[\boldsymbol{e}(0),\,\dot{\bs{e}}(0)]^{\top}$ (thus $\bs{s}(0)$), the objective of the current study is to synthesize an adaptive, generalizable and constrained STC policy $\text{sat}(\bs{u}(t))$ that drives the sliding mode $\bs{s}(t)$ to a predefined bound in a predefined time in the presence of state-dependent perturbation $\bs{\rho}(t,\boldsymbol{x})$ satisfying (\ref{rho_bound}) under limited control authority (\ref{sat_fun_defn}).

\section{Main Result}
\label{section_3}
This section develops the proposed control policy and establishes prescribed-time convergence of the FOSM (\ref{sm_dyn}) under input constraints (\ref{sat_fun_defn}) and state-dependent uncertainty $\bs{\rho}(t,\bs{x})$.

\begin{figure}
    \centering
    \includegraphics[width=0.6\linewidth]{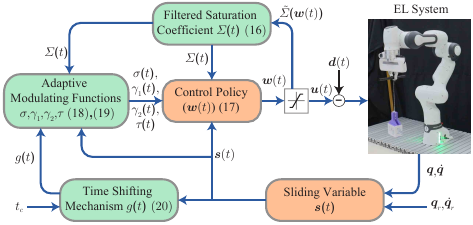}
    \caption{Illustration of the generalized multivariable super-twisting control architecture}
    \label{fig:control_architecture}
\end{figure}
\subsection{Feasibility Analysis}
For system (\ref{sm_dyn}), it may not always be possible to achieve prescribed time convergence from an arbitrarily large initial offset $||\bs{s}(0)||$ given limited control authority $\ol{\boldsymbol{u}}$. Thus, we first derive a sufficient condition that guarantees prescribed time convergence of the constrained system (\ref{sm_dyn}). To this end, we have the following Lemma.

\emph{Lemma 2:} Consider system (\ref{sm_dyn}) under input constraints (\ref{sat_fun_defn}). Then, exact sliding mode convergence is achieved in a user-prescribed convergence time $t_c$ provided the following inequality is satisfied:
\begin{eqnarray}
    \label{feasibility_cond}
    u^*\geq\frac{s^{*}}{\ul{m} t_c}+\frac{\ol{\rho}}{\ul{m}}\sup_t\{\rho_0(\bs{x}(t))\},\,s^{*}=\max_{\bs{s}(0)\in\mathcal{D}}\{||\bs{s}(0)||\},
\end{eqnarray}
where $\ul{m}=\ol{M}^{-1}$, $\mathcal{D}=\{\bs{s}|\,||\bs{s}||\leq\eta\ul{m}t_c u^{*}\}$, $u^*=\text{min}\{\ol{u}_1,..,\ol{u}_n\}$, and $0<\eta<1$.

\emph{Proof:} Consider the Lyapunov function $W=||\bs{s}(t)||$. Taking its time-derivative and using (\ref{sm_dyn}), we have,
\begin{eqnarray}
    \label{Lyap_deriv_1}
    \dot{W}(t){=}\dot{\bs{s}}^{\top}\lceil\bs{s}\rfloor^0{=}\text{sat}(\bs{u})^{\top}(\bs{M}_0^{-1})^{\top}\lceil\bs{s}\rfloor^0{+}(\lceil\bs{s}\rfloor^{0})^{\top}\boldsymbol{\rho}(t,\bs{x}).
\end{eqnarray}
Applying the discontinuous control policy $\text{sat}(\bs{u}){=}{-}u^*\lceil\bs{s}\rfloor^0$,
\begin{eqnarray}
    \label{Lyap_deriv_2}
    \dot{W}(t)&\leq&-u^*(\lceil\bs{s}\rfloor^{0})^{\top}(\bs{M}_0^{-1})^{\top}\lceil\bs{s}\rfloor^0{+}||\boldsymbol{\rho}(t,\bs{x})||\nonumber\\
    &\leq&{-}u^*\ul{m}{+}\rho_0(\bs{x})\ol{\rho}
    {\leq}{-}\eta u^*\ul{m}{-}(1{-}\eta) u^*\ul{m}{+}\rho_0(\bs{x})\ol{\rho}\nonumber\\
    &\leq&-\eta u^* \ul{m},\,\forall \ol{\rho}\leq(1-\eta)u^*\ul{m}/\sup_t\{\rho_0(\bs{x}(t))\}.
\end{eqnarray}
Clearly, the sliding mode $\bs{s}(t)$ converges to the origin in the prescribed time $t_c$ provided:
\begin{eqnarray}
    \label{ustar_cond}
    u^*\geq||\bs{s}(0)||/\eta \ul{m}t_c,\,\forall \ol{\rho}\leq(1-\eta)u^*\ul{m}/\sup_t\{\rho_0(\bs{x}(t))\}.
\end{eqnarray}
By eliminating $\eta$ from the two inequalities in (\ref{ustar_cond}), the task feasibility condition in (\ref{feasibility_cond}) follows directly. $\hfill \square$

\emph{Remark 1:} Observe that the set $\mathcal{D}$ characterizes the initial conditions from which any
closed-loop trajectory reaches the origin within the prescribed time under control
authority $u^{*}$. Substituting $s^{*}=\eta\ul{m}t_c u^{*}$ into (\ref{feasibility_cond})
gives the explicit non-circular condition
\begin{eqnarray}
    \label{feas_cond_1}
    u^*\geq u_{\text{min}}:=\ol{\rho}\sup_t\{\rho_0(\bs{x}(t))\}/(1-\eta)\ul{m}.
\end{eqnarray}
As $\ol{\rho}\sup_t\{\rho_0(\bs{x}(t))\}$ is unknown under Assumptions 1--2,
(\ref{feas_cond_1}) serves as a design guideline: increasing $t_c$ expands $\mathcal{D}$
without affecting $u_{\text{min}}$, while decreasing $\eta$ reduces $u_{\text{min}}$ at
the cost of a smaller domain of attraction $\mathcal{D}$.

\subsection{Adaptive Multivariable Super-Twisting Control Design}
We now present the design of an adaptive, generalized multivariable super-twisting control policy that adheres to prescribed input constraints. To this end, a nonlinear equality transformation is used to convert the constrained input, $\text{sat}(\bs{u}(t))$, into an unconstrained variable $\bs{w}(t)\in\mathbb{R}^n$ as,
\begin{eqnarray}
 \label{input_trans}
 \text{sat}(\bs{u}(t))=\tilde{\bs{\Sigma}}(\bs{w}(t))\,\bs{w}(t),
\end{eqnarray}
with the coefficient $\tilde{\bs{\Sigma}}(\bs{w})$ being a diagonal matrix ${\tilde{\bs{\Sigma}}}(\bs{w})=\text{diag}[\tilde{{\Sigma}}_i(w_i)]\,(1\leq i \leq n)$ defined as \cite{Shao:2022}
\begin{align}
    \label{sat_fun}
   \tilde{\Sigma}_i(w_i(t))=\begin{cases}
    1,& \text{if}\quad w_i(t)\in[-\ol{u}_i,\ol{u}_i]\\
    \overline{u}_i/|w_i(t)|,& \text{otherwise}.\\  
    \end{cases}
\end{align}
The coefficient $\tilde{\Sigma}_i$ is a continuous function that indicates the extent of saturation of the control input $w_i(t)$. Clearly, we have $0<\tilde{\Sigma}_i(w_i(t))\leq 1$, with a smaller value of $\tilde{\Sigma}_i$ indicating a larger extent to which the $i^{\text{th}}$ control input is saturated. In particular, according to the density property of a real number \cite{Hu:2008}, there exists a positive number ${r}$ satisfying $0<{r}\leq\tilde{\Sigma}_i(w_i)\leq1$.

Next, a low-pass filtered version of this signal $\tilde{\bs{\Sigma}}(t)$, denoted by $\bs{\Sigma}(t)$, is considered for controller synthesis below. To this end, we have the following Lemma.

\emph{Lemma 3 \cite{Keshavan:2026}:} Consider the filtered version of the signal $\tilde{\bs{\Sigma}}(\bs{w}(t))$ obtained as
\begin{eqnarray}
    \label{Sigma_lpf}
    h\dot{\bs{\Sigma}}(t)+\bs{\Sigma}(t)=\tilde{\bs{\Sigma}}(\bs{w}(t)),\,\bs{\Sigma}(0)=\tilde{\bs{\Sigma}}(\bs{w}(0)),
\end{eqnarray}
where $h$ is the filter time-constant chosen such that $0{<}h{<<}1$. Then, the filtered signal $\bs{\Sigma}(t)$ satisfies the inequality $0<r\bs{I}\leq\bs{\Sigma}(t)\leq \bs{I}$. Moreover, we have $\dot{\bs{\Sigma}}(t)\in\mathcal{L}_{\infty}$. 

\emph{Remark 2:} From (\ref{Sigma_lpf}), $\bs{\Sigma}(t)$ emulates
$\tilde{\bs{\Sigma}}(\bs{w}(t))$ increasingly closely as $h\rightarrow0$. For the filtering error
$\bs{\rho}_1(t)=[\tilde{\bs{\Sigma}}(t)-\bs{\Sigma}(t)]\bs{w}(t)=h\dot{\bs{\Sigma}}(t)
\bs{w}(t)$, boundedness of $\text{sat}(\bs{u}(t))$ by $\|\bar{\bs{u}}\|$ together with
$\tilde{\bs{\Sigma}}(t)\geq r\bs{I}>0$ from Lemma 3 gives
$\|\bs{w}(t)\|\leq\|\bar{\bs{u}}\|/r$, whence
$||\bs{\rho}_1(t)||\leq\overline{\rho}_1:=hc_{\Sigma}\|\bar{\bs{u}}\|/r$ with
$c_\Sigma=\sup_t\|\dot{\bs{\Sigma}}(t)\|$, an unknown constant vanishing as
$h\rightarrow0$.

Then, by noting that $\bs{\Sigma}(t)$ is an available signal for feedback, the generalized multivariable super-twisting control policy is now constructed in the following form: 
\begin{eqnarray}
\label{control_policy}
\bs{w}(t){=}{-}\rho_0\bs{M}_0{\bs{\Sigma}}_1^{\top}\biggl[\gamma_{1}\lceil{{\bs{s}}} \rfloor^{\alpha} {+} \int_{0}^{t} \gamma_{2}\rho_0\bs{\Sigma}_M\lceil{{\bs{s}}} \rfloor^{\beta} dt{+}\bs{\tau}\biggr],
\end{eqnarray}
with $\frac{1}{2}<\alpha<1,\,\beta=2\alpha-1$, and $\bs{\Sigma}_M(t)={\bs{\Sigma}}_1(t){\bs{\Sigma}}_1(t)^{\top}$ with ${\bs{\Sigma}}_1(t)=\bs{M}_0^{-1}\bs{\Sigma}(t)\bs{M}_0$. The adaptive gains are given by,
\begin{eqnarray}
\gamma_{1} (t) &=& \gamma_{10}{\sigma}^{\alpha} , \gamma_{2} (t) = \gamma_{20} \sigma^{2\alpha}, \nonumber\\
\gamma_{10} &>& \beta\sqrt{\frac{\gamma_{20}}{\alpha}} ,  \gamma_{20} > 0.
\label{eq:gamma}
\end{eqnarray}
\noindent The adaptive modulation functions $\sigma(t)$ and ${\tau}(t)$ are given by,
\begin{eqnarray}
\label{modulating_gain}
\dot{\sigma}(t) &=& \begin{cases} 
    \dfrac{\rho_0\|\bs{s}\|(\eta_1 + \eta_2\|\bs{\zeta}_1\|)}{g(t) - \|\bs{s}\|}\,\sigma, 
    & \text{if } \sigma(t) \geq \sigma_0, \\[10pt] 
    0, & \text{if } \sigma(t) < \sigma_0,
    \end{cases},\nonumber\\
\bs{\tau}(t)&{=}&\frac{\dot{\sigma}}{\rho_0(\bs{x})\sigma}\bs{\Sigma}_M^{-1}\,\bs{s}(t),
\end{eqnarray}
where $\sigma_0 = \sigma(0) > 0$, $\eta_1,\,\eta_2>0$, $\bs{\zeta}_1 = \sigma^\alpha\lfloor\bs{s}\rceil^\alpha$, and 

\begin{eqnarray}
\label{modulating_gain_2}
g(t)&=&\varepsilon/{\nu}(t),\,\varepsilon> 0,\nonumber\\
    \nu(t)&=&\begin{cases}
    \frac{1}{2}(1-\cos(\frac{\pi (t-t^k_1)}{t_c})),\,\,\,\,\,\,\,\text{if}\,\,\,\,\,t<t_c+t^k_1\\
    1,\,\,\,\,\,\qquad\qquad\qquad\qquad\text{if}\,\,\,\,\,t\geq t_c+t^k_1,\\  
    \end{cases}
\end{eqnarray}
where $t^k_1$ is the time instant corresponding to the $k^{\text{th}}$ instance when $\nu(t)||\bs{s}(t)||<\varepsilon$ is violated, and $t_c$ denotes the predefined convergence time.

\emph{Remark 3}: Note that $g(t)$ is designed to be a continuous monotonically decreasing function that takes an initial value of $g(0)=+\infty$ and attains a terminal value of $g(t)=\varepsilon\,\forall\,t\geq t_c$. Thus we have $||\bs{s}(0)||<g(0)$ for any arbitrary initial condition $\bs{s}(0)$, and the size of the envelope that constrains the evolution of $\bs{s}(t)$ is parameterized by the variables $\varepsilon,\,t_c$ that may be arbitrarily chosen for practical applications. However, the main drawback of previous barrier-function based approaches in \cite{SPSTC,Obeid:2020} is that the stability of the control system is compromised when $\nu||\bs{s}||\geq\varepsilon$, which possibly arises due to a sudden change in the external disturbance magnitude, so that system trajectory never returns to this bound once escaping from it. Thus, prior approaches suffer from a lack of global stability guarantee. In order to overcome this drawback, this study incorporates the time-shifting scheme in (\ref{modulating_gain_2}) which ensures that the barrier function $||\bs{s}(t)||/(g(t)-||\bs{s}(t)||)$ remains well-defined $\forall\,t\geq 0$.

\subsection{Stability analysis}
Under the proposed control scheme \eqref{control_policy}-\eqref{modulating_gain_2}, the stability analysis of system \eqref{sm_dyn} is now undertaken as the proof of the following theorem.

\begin{theorem}
    Consider the nonlinear system \eqref{sm_dyn} subject to Assumptions 1-2 and input saturation \eqref{input_trans}. Further, assume that sufficient control authority exists such that the task feasibility condition (\ref{feasibility_cond}) is satisfied. Then, with the controller gains chosen such that \eqref{eq:gamma} holds, the control policy \eqref{control_policy} with the adaptive gain update given by \eqref{modulating_gain}-(\ref{modulating_gain_2}) ensures that the sliding variable $\bs{s}(t)$ converges in a prescribed-time $t_c$ to an ultimate bound $\varepsilon>0$, so that $||\bs{s}(t)||<\varepsilon\,\forall\,t\geq t_c+t_1^k$.
\end{theorem}

\emph{Proof}: By  substituting (\ref{input_trans}) in (\ref{sm_dyn}), we have,
\begin{eqnarray}
\label{FOS_00}
\dot{\bs{s}}(t)=\bs{M}_0^{-1}\bs{\Sigma}(t)\bs{w}(t)+\tilde{\bs{\rho}}(t,\boldsymbol{x}),
\end{eqnarray}
where $\tilde{\bs{\rho}}(t,\boldsymbol{x})=\bs{M}_0^{-1}\bs{\rho}_1(t)+\bs{\rho}(t,\bs{x})$. 

By substituting $\bs{w}(t)$ from \eqref{control_policy}  in (\ref{FOS_00}), the closed-loop dynamics can be written as,
\begin{eqnarray}
\dot{\bs{s}}{=}{-}\rho_0\bs{\Sigma}_M\biggl[\gamma_{1}\lceil{\bs{s}} \rfloor^{\alpha} {+} \int_{0}^{t} \gamma_{2} \rho_0\bs{\Sigma}_M\lceil{\bs{s}} \rfloor^{\beta} dt{+}\bs{\tau}{-}\frac{\bs{\Sigma}_M^{-1}\tilde{\bs{\rho}}}{\rho_0}\biggr].
\label{FOS_1}
\end{eqnarray}

The stability of this closed-loop system is now studied for the two possible cases of $\nu(t)||\bs{s}(t)||<\varepsilon$ and $\nu(t)||\bs{s}(t)||\geq\varepsilon$. 

\emph{Case (a):} For $\nu(t)||\bs{s}(t)||{<}\varepsilon$, consider the change of variable:
\begin{eqnarray}
    \label{zeta_defn}
\boldsymbol{\zeta}{=}\left[\begin{array}{c}
\bs{\zeta}_1\\
\bs{\zeta}_2\end{array}\right]{=}\left[\begin{array}{c}
\sigma^{\alpha}\lceil{\bs{s}}\rfloor^{\alpha}\\
{-}\int_{0}^t\,\gamma_2(t)\rho_0(\boldsymbol{x})\bs{\Sigma}_M\lceil{\bs{s}}\rfloor^{\beta}\,dt\end{array}\right].
\end{eqnarray}
Now, by invoking the fact that $||\bs{s}||^{\alpha}=(\bs{s}^{\top}\bs{s})^{\alpha/2}$ , we have $\frac{d||\bs{s}||^{\alpha}}{dt}=\alpha\bs{s}^{\top}\dot{\bs{s}}/||\bs{s}||^{2-\alpha}$. Then, using \eqref{FOS_1}, the time-derivative of $\bs{\zeta}_1$ can be written as
\begin{eqnarray}
    \label{zeta1_deriv_2}
    \dot{\bs{\zeta}}_1=\alpha\frac{\dot{\sigma}}{\sigma}\bs{\zeta}_1{+}\frac{\sigma}{||\bs{\zeta}_1||^{\frac{1-\alpha}{\alpha}}}\biggl[\bs{I}-(1-\alpha)\frac{\bs{\zeta}_1\bs{\zeta}_1^{\top}}{\bs{\zeta}_1^{\top}\bs{\zeta}_1}\biggr]\dot{\bs{s}},
\end{eqnarray}
where we have invoked the fact that $||\bs{s}||^{\alpha-1}=\sigma^{1-\alpha}/||\bs{\zeta}_1||^{\frac{1-\alpha}{\alpha}}$. From \eqref{modulating_gain}, we have, 
\begin{eqnarray}
    \label{tau_defn}
\bs{\tau}{=}\frac{\dot{\sigma}}{\rho_0\sigma}\bs{\Sigma}_M^{-1}\bs{s}{=}\frac{\dot{\sigma}}{\rho_0\sigma^2}\bs{\Sigma}_M^{-1}\lceil\bs{\zeta}_1\rfloor^{1/\alpha}{=}\frac{\dot{\sigma}}{\rho_0\sigma^2}\bs{\Sigma}_M^{-1}||\bs{\zeta}_1||^{\frac{1-\alpha}{\alpha}}\bs{\zeta}_1.
\end{eqnarray}
Using (\ref{tau_defn}), we have $\bigl[\bs{I}-\bs{\zeta}_1\bs{\zeta}_1^{\top}/\bs{\zeta}_1^{\top}\bs{\zeta}_1\bigr]\bs{\Sigma}_M\,\bs{\tau}=\dot{\sigma}\bs{\zeta}_1/\sigma-\sigma\rho_0\bs{\Sigma}_M\bs{\tau}/||\bs{\zeta}_1||^{\frac{1-\alpha}{\alpha}}=\bs{0}$. Then, substituting \eqref{tau_defn} and (\ref{FOS_1}) in \eqref{zeta1_deriv_2}, it follows that
\begin{eqnarray}
    \label{zeta1_deriv_1}    \dot{\bs{\zeta}}_1=\frac{\alpha\sigma \rho_0(\boldsymbol{x})\bs{\Sigma}_M}{||\bs{\zeta}_1||^{\frac{1-\alpha}{\alpha}}}\biggl[-\gamma_{10}\bs{\zeta}_1{+}\bs{\zeta}_2\biggr]{+}\bs{\rho}_3(t,\bs{\zeta}),
\end{eqnarray}
where
\begin{eqnarray}
    \label{rho1_defn}
    \bs{\rho}_3(t,\bs{\zeta})=\frac{\sigma}{||\bs{\zeta}_1||^{\frac{1-\alpha}{\alpha}}}\biggl(\underbrace{\biggl\{\alpha\bs{I}+(1-\alpha)\biggl[\bs{I}-\frac{\bs{\zeta}_1\bs{\zeta}_1^{\top}}{\bs{\zeta}_1^{\top}\bs{\zeta}_1}\biggr]\biggr\}\tilde{\bs{\rho}}}_{\bs{\rho}_2(t,\bs{\zeta})}\nonumber\\{+}\rho_0(1-\alpha)\biggl\{\biggl[\bs{I}-\frac{\bs{\zeta}_1\bs{\zeta}_1^{\top}}{\bs{\zeta}_1^{\top}\bs{\zeta}_1}\biggr]\biggl[-\gamma_{10}\bs{\Sigma}_M\bs{\zeta}_1{+}\bs{\Sigma}_M\bs{\zeta}_2\biggr]\biggr\}\biggr).
\end{eqnarray}
As $||\bs{I}-\bs{\zeta}_1\bs{\zeta}_1^{\top}/\bs{\zeta}_1^{\top}\bs{\zeta}_1||\leq 2$ and $\rho_0(\bs{x})\geq 1$, we have $||\bs{\rho}_2(t,\bs{\zeta})||\leq (2-\alpha)\rho_0(\bs{x})\overline{\rho}_2$, where $\overline{\rho}_2(t)=\overline{\rho}+\ol{m}\sup_t\{||\bs{\rho}_1(t)||\}$. The time-derivative of $\bs{\zeta}_2$ is now obtained as:
\begin{eqnarray}
    \label{zeta2_deriv}
    \dot{\bs{\zeta}}_2{=}{-}\gamma_{20}\sigma^{2\alpha}\rho_0(\boldsymbol{x})\bs{\Sigma}_M\lceil \bs{s} \rfloor^{\beta}{=}{-}\gamma_{20}\frac{\sigma\bs{\Sigma}_M \rho_0(\boldsymbol{x})}{||\bs{\zeta}_1||^{\frac{1-\alpha}{\alpha}}}\bs{\zeta}_1.
\end{eqnarray}
Now, stacking \eqref{zeta1_deriv_1}-\eqref{zeta2_deriv}, we have,
\begin{eqnarray}
\label{zeta_dyn_1}
\dot{\boldsymbol{\zeta}}{=} \frac{\sigma\rho_0}{||\bs{\zeta}_1||^{\frac{1{-}\alpha}{\alpha}}}\biggl\{\underbrace{\left[\begin{array}{cc}
{-}\alpha\gamma_{10} & \alpha\\
{-}\gamma_{20} & 0\end{array}\right]\otimes \bs{\Sigma}_M }_{\boldsymbol{A}}\boldsymbol{\zeta}{+}\underbrace{\left[\begin{array}{c}
\bs{I}\\
\bs{O}\end{array}\right]}_{\boldsymbol{C}^{\top}}\frac{\bs{\rho}_2}{\rho_0}\nonumber\\
{+}\underbrace{(1{-}\alpha)\biggl[\bs{I}{-}\frac{\bs{\zeta}_1\bs{\zeta}_1^{\top}}{\bs{\zeta}_1^{\top}\bs{\zeta}_1}\biggr]\left[\begin{array}{cc}
{-}\gamma_{10} & 1\\
0 & 0\end{array}\right]\otimes \bs{\Sigma}_M }_{\boldsymbol{B}(t,\bs{\zeta})}\boldsymbol{\zeta}\biggr\},
\end{eqnarray}
where $\bs{O}$ is a square matrix of dimension $n$. Clearly, system \eqref{zeta_dyn_1} retains the structure of the STA while accounting for the presence of the saturation constraint $\bs{\Sigma}(t)$ and the state-dependent perturbation $\rho_0(\boldsymbol{x})$. To verify the stability of system \eqref{zeta_dyn_1}, consider the Lyapunov function candidate $V(\boldsymbol{\zeta})=\sigma^{-1}\boldsymbol{\zeta}^{\top}\boldsymbol{P}\boldsymbol{\zeta}$, where $\boldsymbol{P}$ is a constant symmetric positive definite matrix that is defined as
\begin{eqnarray}
\label{P_defn}
\boldsymbol{P}=\left[\begin{array}{cc}
\gamma_{10}+\frac{4 \gamma_{20}}{\gamma_{10}} & -1\\
-1 & \frac{2}{\gamma_{10}}\end{array}\right]\otimes\bs{I}.
\end{eqnarray}
Now, by design, $\bs{M}_0$ is a symmetric positive definite matrix, thus owing to a similarity transformation, ${\bs{\Sigma}}_1=\bs{M}_0^{-1}\bs{\Sigma}\bs{M}_0$ is a positive definite matrix, with $\lambda_{\text{min}}\{{\bs{\Sigma}_1}\}=\lambda_{\text{min}}\{{\bs{\Sigma}}\}\geq r$. It then follows that $\bs{\Sigma}_M={\bs{\Sigma}}_1{\bs{\Sigma}}_1^{\top}$ is also a symmetric positive definite matrix as ${\bs{\Sigma}}_1$ is invertible, thus there exists a constant $\omega_M$ satisfying $\omega_M=\lambda_{\text{min}}\{\bs{\Sigma}_M\}> 0$. Further, as $||\bs{\Sigma}_M||\leq||{\bs{\Sigma}}_1||^2\leq \ol{m}^2\ol{M}^2$,it is straightforward to conclude that $||\bs{B}||{\leq}\omega_B$, where $\omega_B{=}2(1{-}\alpha)\sqrt{1{+}\gamma_{10}^2} \ol{m}^2\ol{M}^2$.

The time-derivative of the Lyapunov function then satisfies:
\begin{eqnarray}
\label{V_deriv_eqn}
\dot{V}
{\leq}\frac{1}{||\bs{\zeta}_{1}||^{\frac{1-\alpha}{\alpha}}}\biggl[\rho_0\biggl\{-\boldsymbol{\zeta}^{\top}\boldsymbol{Q}\boldsymbol{\zeta}
{+}2 q_1 \overline{\rho} \,||\boldsymbol{\zeta}||{+}2q_2 ||\bs{\zeta}||^2\biggl\}\biggr]\nonumber\\{-}\rho_0\,\biggl[\frac{\eta_1\bigl||\bs{s}||}{(g(t)-||\bs{s}||)\sigma}{+}\frac{\eta_2\bigl||\bs{s}||\,||\bs{\zeta}_1||}{(g(t)-||\bs{s}||)\sigma}\biggr]\boldsymbol{\zeta}^{\top}\boldsymbol{P}\boldsymbol{\zeta},
\end{eqnarray}
where $q_1=(2-\alpha)||\boldsymbol{C}\boldsymbol{P}||$, $q_2=2\omega_B||\bs{P}||$, and 
\begin{eqnarray}
\boldsymbol{Q}&=&{-}\boldsymbol{A}^{\top} \boldsymbol{P}{-}\boldsymbol{P}\boldsymbol{A}=\bs{Q}_0\otimes\bs{\Sigma}_M,\nonumber\\\bs{Q}_0&=&\left[\begin{array}{cc}
2\alpha\gamma_{10}^2{+}8\alpha\gamma_{20}{-}2\gamma_{20} & {-}2\beta\frac{\gamma_{20}}{\gamma_{10}}{-}2\alpha{\gamma_{10}}\\
{-}2\beta\frac{\gamma_{20}}{\gamma_{10}}{-}2\alpha{\gamma_{10}} & 2\alpha\end{array}\right].\,\,\,\,\,\,
\end{eqnarray}
Then, for the choice of gains given by \eqref{eq:gamma}, it is straightforward to verify the following inequalities:
\begin{eqnarray}
\label{Sigma_eqns}
\Pi_1{=}2\alpha\gamma_{10}^2{+}8\alpha\gamma_{20}{-}2\gamma_{20}>0,\,
\Pi_2{=}\alpha\gamma_{20}{-}\beta^2\frac{\gamma_{20}^2}{\gamma_{10}^2}{>}0.
\end{eqnarray}
$\Pi_1$ and $\Pi_2$ are the principal minor determinants of $\boldsymbol{Q}_0$ respectively, so that from \eqref{Sigma_eqns}, one can conclude that the matrix $\boldsymbol{Q}_0$ is positive-definite. Since $\bs{\Sigma}_M$ is positive definite (shown above), it follows from Lemma 1 that $\bs{Q}=\bs{Q}_0\otimes\bs{\Sigma}_M$ is also positive definite. Thus, we have,
\begin{eqnarray}
\label{V_deriv_eqn_2}
\dot{V}{\leq}\frac{\rho_0}{||\bs{\zeta}_{1}||^{\frac{1-\alpha}{\alpha}}}\biggl[\biggl\{2 q_1 \overline{\rho}_2||\boldsymbol{\zeta}||{+}2q_2 ||\bs{\zeta}||^2{-}\lambda_{\text{min}}\{\boldsymbol{Q}\}||\boldsymbol{\zeta}||^2\biggl\}\biggr]\nonumber\\{-}\rho_0\biggl[\frac{\eta_1\bigl||\bs{s}||}{(g-||\bs{s}||)\sigma}{+}\frac{\eta_2||\bs{s}||\,||\bs{\zeta}_1||}{(g-||\bs{s}||)\sigma}\biggr]\lambda_{\text{min}}\{\boldsymbol{P}\}||\boldsymbol{\zeta}||^2.
\end{eqnarray}
Then, with $\omega_Q=\omega_M\lambda_{\text{min}}\{\boldsymbol{Q}_0\},\,\omega_P=\lambda_{\text{min}}\{\boldsymbol{P}\}>0$, and invoking the definition of $\bs{\zeta}_1$ in (\ref{zeta_defn}), we have $\sigma^{-1}=||\bs{s}||/||\bs{\zeta}_1||^{1/\alpha}$, so that from \eqref{V_deriv_eqn_2},
\begin{eqnarray}
    \label{V_deriv_eqn_200}
    \dot{V}{\leq}{-}\frac{\omega_Q\rho_0}{||\bs{\zeta}_{1}||^{\frac{1{-}\alpha}{\alpha}}}||\boldsymbol{\zeta}||^2
    {-}\frac{\omega_P\rho_0}{||\bs{\zeta}_{1}||^{\frac{1{-}\alpha}{\alpha}}}\biggl[
\eta_1\biggl\{\frac{\bigl||\bs{s}||^2}{g{-}||\bs{s}||}{-}\mu_1\biggr\} ||\boldsymbol{\zeta}||\nonumber\\{+}\eta_2\biggl\{
\frac{\bigl||\bs{s}||^2}{g{-}||\bs{s}||}{-}\mu_2\biggr\} ||\boldsymbol{\zeta}||^2\biggr],
\end{eqnarray}
where $\mu_1=2q_1 \overline{\rho}_2/\eta_1\,\omega_P,\,\mu_2=2q_2/\eta_2\,\omega_P>0$. Now, observe that by definition we have $\rho_0(\boldsymbol{x})\geq 1$. Thus, $\dot{V}\leq 0$ if $\frac{\bigl||\bs{s}||^2}{g(t)-||\bs{s}||}{-}\mu\geq0$ where $\mu=\text{max}\{\mu_1,\,\mu_2\}$, or equivalently, $||\bs{s}||^2+\mu||\bs{s}||-\mu g(t)\geq0$  which is a quadratic inequality in $||\bs{s}||$. Solving this inequality, it is straightforward to conclude that $\dot{V}\leq0$ provided $||\bs{s}||\geq\varpi$, where   
\begin{eqnarray}
    \label{uub_defn}
    \varpi(t)=\frac{-\mu+\sqrt{\mu^2+4\mu g(t)}}{2}.
\end{eqnarray} 

Since $\sigma(t) \geq \sigma_0 > 0\ \forall\, t \geq 0$ by \eqref{modulating_gain},
$V(t)$ remains well-defined for all $t \geq 0$, and is bounded whenever
$\|\bs{s}\| \geq \varpi(t)$. Moreover, $\mu^2 + 4\mu g(t) < (2g(t)+\mu)^2$
implies $\varpi(t) < g(t)\ \forall\, t \geq 0$. 

\noindent
We now show $\|\bs{s}(t)\| < g(t)$
for all $t \geq 0$ by contradiction. By construction, we have $||\bs{s}(0)||<g(0)=+\infty$. Then, suppose that there exists a finite
$T > 0$ with $\|\bs{s}(T)\| = g(T)$ and $\|\bs{s}(t)\| < g(t)$ on $[0,T)$.
Then $\lim_{t \to T}\dot{\sigma}/\sigma = +\infty$, so that from (\ref{modulating_gain}), it follows that $\lim_{t\rightarrow T}\sigma(t)=+\infty$.

Since $\|\bs{\zeta}_1\| =
\sigma^{\alpha}\|\bs{s}\|^{\alpha}$ from \eqref{zeta_defn}:
\begin{equation}
    \label{eq:V_lower}
    V = \sigma^{-1}\bs{\zeta}^{\top}\bs{P}\bs{\zeta}
      \geq \omega_P\sigma^{2\alpha-1}\|\bs{s}\|^{2\alpha}.
\end{equation}

For $\alpha \in (1/2, 1)$, we have $2\alpha-1>0$, so $\lim_{t\to T}V(t)=+\infty$.
Since $\varpi(T)<g(T)=\|\bs{s}(T)\|$ and both functions are 
continuous, there exists $\delta>0$ such that $\|\bs{s}(t)\|>\varpi(t)$ 
for all $t\in(T{-}\delta,T)$, so \eqref{V_deriv_eqn_200} gives 
$\dot{V}(t)\leq 0$ on $(T{-}\delta,T)$. A non-increasing $V$ 
cannot diverge to $+\infty$, yielding a contradiction. Thus, no such $T$ exists, hence $\|\bs{s}(t)\|<g(t)\ 
\forall\,t\geq 0$. Since $g(t)=\varepsilon$ for $t\geq t_c$, 
we conclude $\|\bs{s}(t)\|<\varepsilon\ \forall\,t\geq t_c$.

\emph{Case (b):} At the instant $t=t^k_1$, we have $\nu(t)||\bs{s}(t)||\geq\varepsilon$, so that $g(t^k_1)=\varepsilon/\nu(t^k_1)$, where  according to (\ref{modulating_gain_2}), $\nu(t^k_1)=0$, so that following Remark 3, $g(t^k_1)=+\infty$ and $g(t)=\varepsilon\forall\,t\geq t_c+t^k_1$. The stability analysis of the closed-loop under the proposed controller (\ref{control_policy}) then follows case (a) $\forall t\geq t^k_1$, so that $||\bs{s}(t)||<\varepsilon\,\forall\,t\geq t_c+t^k_1$. The convergence time is independent of initial conditions due to the time-varying barrier function design, ensuring prescribed-time convergence in the sense of Definition 1. This ends the proof. $\hfill \square$

\emph{Remark 4:} A less conservative estimate of the steady-state ultimate bound $\varpi(t)$ in (\ref{uub_defn}) can be obtained as $\varpi^*:=\varpi|_{g=\varepsilon}=\frac{-\mu+\sqrt{\mu^2+4\mu \varepsilon}}{2}$. Then for a sufficiently large choice of gains $\eta_1,\,\eta_2$, we have $\mu=\text{max}\{2 q_1\ol{\rho}_2/\eta_1\omega_p,\,2q_2/\eta_2\omega_P\}<<1$, so that $\varpi^*=\sqrt{\mu\varepsilon}-\frac{\mu}{2}+\mathcal{O}(\mu^{3/2})\approx\sqrt{\mu\varepsilon}$ as $\mu\rightarrow0$. Thus, to achieve a prescribed bound $\varpi^*$, the design parameter $\varepsilon$ should satisfy $\varepsilon=\varpi^{*2}/\mu$, so that this bound then depends only on gain parameters $\eta_1,\,\eta_2,\,\alpha,\omega_P$ (thus $\gamma_{10},\,\gamma_{20}$), and system parameters $\rho_2,\,q_1,\,q_2$, and not on initial conditions $\bs{s}(0)$ or prescribed-time $t_c$. Thus, the designer is free to select $t_c$ to expand the domain of attraction $\mathcal{D}$ (in accordance with Remark 1) without altering achievable steady-state accuracy. 

\emph{Remark 5:} The proposed policy \eqref{control_policy} relies on the gain 
adaptation mechanism \eqref{modulating_gain} to ensure prescribed-time convergence 
of $\bs{s}$ to the ultimate bound $\varepsilon$. However, from \eqref{modulating_gain}, it is apparent that the adaptive gain $\sigma$ is monotonically increasing, which may lead to a fairly large control action (i) even in steady-state when $0<||\bs{s}||<\varepsilon$, or (ii) even after the control input $\bs{w}(t)$ enters the saturation regime $\bs{w}\in\mathcal{W}_S$, where $\mathcal{W}_S=\{\bs{w}\in\mathbb{R}^n:\bs{\Sigma}(\bs{w})<\bs{I}\}$ denotes the saturated set, and $\mathcal{W}_{N}$ denotes the unsaturated set $\mathcal{W}_{N}=\{\bs{w}\in\mathbb{R}^n:\bs{\Sigma}(\bs{w})=\bs{I}\}$. In order to mitigate this issue, we introduce a dead-zone technique as follows \cite{Utkin:2013}:
\begin{align}
\label{dead_zone}
\dot{\sigma}(t) = \begin{cases}
\frac{\rho_0\|\bs{s}\|(\eta_1{+}\eta_2\|\bs{\zeta}_1\|)}{g(t){-}\|\bs{s}\|}\sigma, 
& \text{if } \|\bs{s}\| \geq \varepsilon,\, \bs{w}\in\mathcal{W}_N,\, 
\sigma {\geq} \sigma_0 \\
\frac{\rho_0\|\bs{s}\|(\eta_1{+}\eta_2\|\bs{\zeta}_1\|)}{g(t){-}\|\bs{s}\|}\sigma, 
& \text{if } \|\bs{s}\| \geq \varepsilon,\, \bs{w}\in\mathcal{W}_S,\, 
\sigma(t_2){=}\sigma(0) \\
0, & \text{otherwise},
\end{cases}
\end{align}
where $t_2$ is the time-instant when $\bs{w}(t)$ enters the saturation regime $\mathcal{W}_S$. Clearly, system (\ref{dead_zone}) ensures that $\sigma$ will not increase when $||\bs{s}||<\varepsilon$ or $\bs{u}(t)$ is saturated, thereby preventing unnecessary growth of the control gain. 
Under \eqref{dead_zone}, the reset $\sigma(t_2)=\sigma_0$ ensures
$\sigma(t)\geq\sigma_0>0$ for all $t\geq0$, preserving the
well-definedness of $V$. Moreover, whenever $||\bs{s}||\geq\varepsilon$
and $\bs{w}\in\mathcal{W}_N$, the adaptation in \eqref{dead_zone}
coincides with that in \eqref{modulating_gain}, so that the analysis of
Case (a) applies verbatim in the unsaturated regime. In the saturated
regime, the control action is limited by the actuation bound and the
prescribed envelope may be violated; the resulting escape is handled by
Case (b), which reinitializes $g(t)$ at the corresponding reset instant
and guarantees re-convergence to $\varepsilon$ within $t_c$ thereafter.

\emph{Remark 6:} At $\alpha = 1/2$, $\beta = 0$, the integral term in \eqref{control_policy} reduces to the discontinuous sign-function $\lceil \bs{s} \rfloor^0$ of \cite{tian2019adaptive,Obeid:2020}, reintroducing residual chattering; this case lies outside the admissible set (\ref{eq:gamma}). The restriction $\alpha \in (1/2, 1)$ ensures $\beta = 2\alpha - 1 > 0$, rendering the integral term continuous and eliminating this drawback, which is the primary motivation for the proposed scheme. Moreover, for any fixed $\alpha \in (1/2, 1)$, the admissible gain set \eqref{eq:gamma} requires only $\gamma_{20} > 0$, which is strictly more relaxed than the condition $\gamma_{20} > 1$ in \cite{Keshavan:2026}.

\emph{Remark 7:} Unlike \cite{Jimenez:2020}, note that the proposed scheme achieves prescribed-bound convergence in a user-prescribed convergence time $t_c$ that is completely decoupled from system parameters. Moreover, from Remark 4, it follows that since $\mu_1\propto \ol{\rho}_2$, the bound scales sublinearly with the perturbation bound as $\varpi^*\approx \sqrt{\ol{\rho_2}\varepsilon}$ which is more benign than the fixed-gain scheme in \cite{Mei:2023} that achieves linear scaling $\varpi^*\propto\ol{\rho}_2$. It also differs from the prescribed-time framework in \cite{Song:2017,Song:2019} in that convergence is to the ultimate bound $\varpi^*$ rather than the origin, and control gain $\sigma(t)$ remains bounded for all $t\geq 0$ via the dead-zone modification (\ref{dead_zone}), avoiding the infinite-gain singularity inherent to time-varying gain-based prescribed-time methods.

\section{Results and Discussion}\label{section4}
This section presents detailed experimental results on the Franka Research~3~(FR3)
arm, a 7-DoF serial manipulator, with comparisons against
designs~\cite{SPSTC}, \cite{bertingo_7dof}, and \cite{Moreno:2022}.

\begin{figure}
    \centering
    \includegraphics[width=0.6\linewidth]{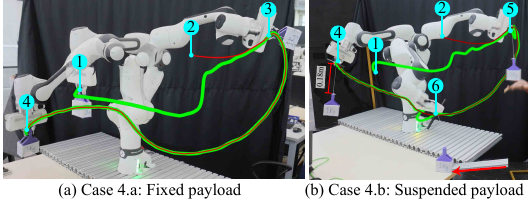}
    \caption{FR3 experimental setup: reference (red) and traced (green) end-effector paths for the $30^{\circ}$ initial-offset tracking task. The points (1--6) indicate: 1. robot's start position, 2. reference trajectory's start position, 3. trajectory error convergence at prescribed time $t_c$, 4. end of tracking, 5. release of payload, and 6. external perturbation to the load.}
    \label{fig:fr3_experimental_setup}
\end{figure}

\subsection{Experimental Validation}
The following cases are considered.
\begin{enumerate}[label =$\bullet$, ref=\arabic*]
    \item Case 1: Tracking of a minimum-jerk reference trajectory~\cite{Kyriakopoulos_Saridis_1994} from a $30^\circ$ initial offset, with $t_c=6$\,s, $\alpha=0.7$, and no external disturbance.\label{case:tracking}
    \item Case 2: As in Case~\ref{case:tracking}, with a $20^\circ$ reference trajectory reset applied to each joint at $t=13$\,s.\label{case:trajectory_reset}
    \item Case 3: Variation of $\bs{x}(0)$, $t_c$, and $\alpha$.
    \begin{enumerate} [label =$\bullet$, ref=3.\alph*]
        \item Case 3.a: Joint offsets of $10^\circ$, $20^\circ$, $30^\circ$ in $(|q_i-q_{i,r}|)$, with $\alpha = 0.7$, $t_c = 6$\,s.\label{case:initial_conditions}
        \item Case 3.b: $t_c =\{3,4,5,6\}$\,s, with $\alpha=0.7$ and fixed initial position.\label{case:exp_prescribed_time}
        \item Case 3.c: $\alpha=\{0.6,0.7,0.8,0.9\}$, with $t_c=6$\,s.\label{case:alpha}
    \end{enumerate}
    \item Case 4: External perturbation.
    \begin{enumerate} [label =$\bullet$, ref=4.\alph*]
        \item \label{case:payload_conditions} Case 4.a: Fixed payload of $0$, $0.5$, $1$\,Kg attached to the end-effector from $t=0$, with $\alpha=0.7$, $t_c=6$\,s.
        \item Case 4.b: Suspended payloads of $0.5$ and $1$\,Kg on a $0.18$\,m cable, with $\alpha=0.7$, $t_c=4$\,s. Two non-smooth events are applied to each payload: release from slack, at which the cable goes taut and transfers the load impulsively, and a subsequent manual disturbance to the suspended mass.\label{case:suspended_load}
    \end{enumerate}
\end{enumerate}

\begin{table}[htbp]
\centering
     \caption{Steady state error metrics obtained from experiments using the FR3 arm across various case studies.}
\begin{threeparttable}[b]
     \begin{tabular}{
    |>{\centering\arraybackslash}m{3.3em}|
     >{\raggedright\arraybackslash}m{9.9em}|
     >{\centering\arraybackslash}m{2em}|
     >{\centering\arraybackslash}m{2.45em}|
     >{\centering\arraybackslash}m{2em}|
     >{\centering\arraybackslash}m{2.1em}|}
     \hline
     \multirow{2}{*}{Study} & \multirow{2}{*}{Condition} & \multicolumn{2}{c|}{$\norm{\bs{e}}$ (deg)} & \multicolumn{2}{c|}{$\norm{\dot{\bs{e}}}$ (deg/s)} \\\cline{3-6}
     & & MSE\tnote{a} & RMSE\tnote{b} & MSE & RMSE \\\hline
      Case: \ref{case:tracking}& $t_1^0=0$ & 1.1 & 0.3 & 4.5 & 1.9\\\hline
      Case: \ref{case:trajectory_reset} & $t_1^1=13$& 0.4 & 0.1 & 9.4 & 2.4\\\hline
      \multirow{2}{*}{Case: \ref{case:payload_conditions}} & fixed load = 0.5Kg & 1.0 & 0.3 & 4.1 & 1.6 \\\cline{2-6}
      & fixed load = 1Kg & 0.7 & 0.2 & 5.7 & 2.4 \\\hline
      \multirow{2}{*}{Case: \ref{case:suspended_load}} & suspended load  = 0.5Kg& 1.0 & 0.5 & 6.1 & 3.3 \\\cline{2-6}
      & suspended load = 1Kg& 0.6 & 0.3 & 8.7 & 2.2 \\\hline
     \end{tabular}
     \begin{tablenotes}
    \item [a] Maximum Steady State Error (MSE) $=\max_{t\in[18,\,25]} \norm{\bs{e}(t)}$
    \item [b] Root Mean Square Error (RMSE) $= \sqrt{\sum_{t\in[18,\,25]} {\norm{\bs{e}(t)}}^2/{N}}$
    \end{tablenotes}
      \end{threeparttable}
      \label{tab:franka_exp_metrics}
 \end{table}

\begin{figure*}
    \centering
    \subfloat[ Case: \ref{case:tracking}: $\bs{q}(t)$ in rad]{\includegraphics[width=0.333\textwidth]{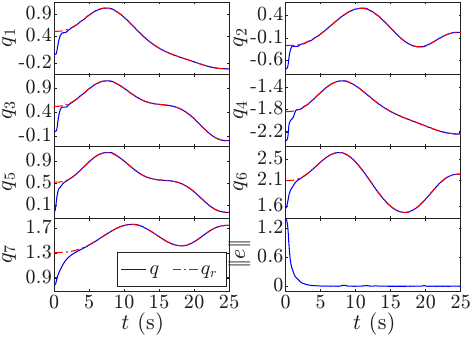}%
    \label{fig:gassta_fr3_jointPosition}}
    \hfill
    \subfloat[ Case: \ref{case:tracking}: $\dot{\bs{q}}(t)$ in rad/s]{\includegraphics[width=0.333\textwidth]{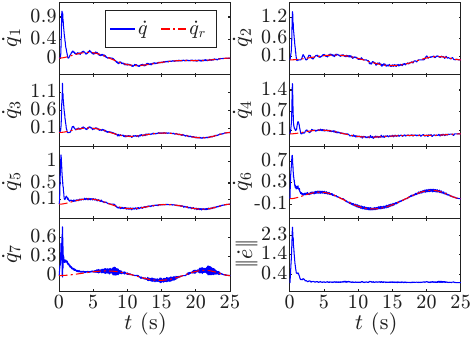}%
    \label{fig:gassta_fr3_jointVelocity}}
    \hfill
    \subfloat[ Case: \ref{case:tracking}: $\bs{u}(t)$ in Nm]{\includegraphics[width=0.333\textwidth]{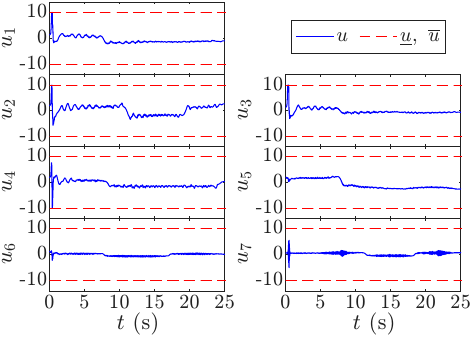}%
    \label{fig:gassta_fr3_input}}
    \hfill
    \subfloat[Case: \ref{case:trajectory_reset}: $\bs{q}(t)$ in rad]{\includegraphics[width=0.333\textwidth]{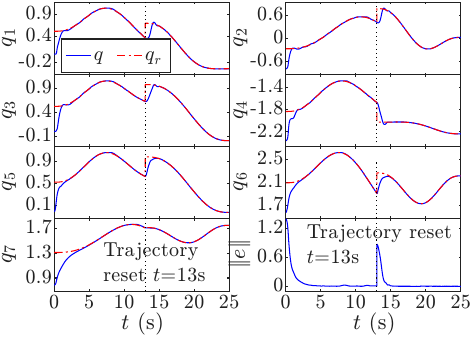}%
    \label{fig:gassta_fr3_reset_jointPosition}}
    \hfill
    \subfloat[Case: \ref{case:trajectory_reset}: $\dot{\bs{q}}(t)$ in rad/s]{\includegraphics[width=0.333\textwidth]{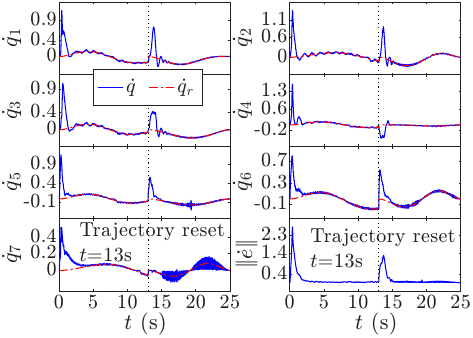}%
    \label{fig:gassta_fr3_reset_jointVelocity}}
    \hfill
    \subfloat[Case: \ref{case:trajectory_reset}: $\bs{u}(t)$ in Nm]{\includegraphics[width=0.333\textwidth]{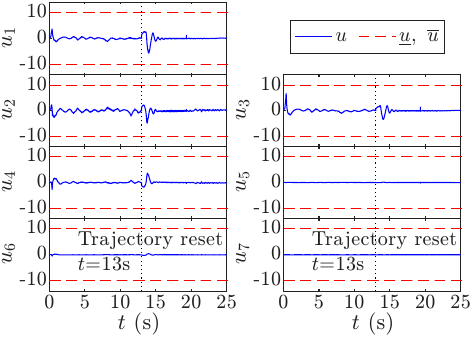}%
    \label{fig:gassta_fr3_reset_input}}
    \hfill
    \subfloat[ Case: \ref{case:tracking}: $\norm{\bs{s}(t)},\sigma(t),\gamma_1(t),\gamma_2(t)$  ]{\includegraphics[width=0.333\textwidth]{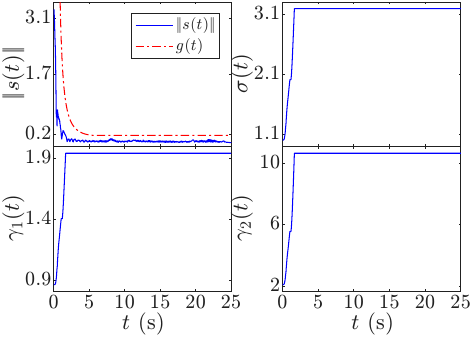}%
    \label{fig:gassta_fr3_sliding_surface}}
    \hfill
    \subfloat[Case: \ref{case:trajectory_reset}: $\norm{\bs{s}(t)},\sigma(t),\gamma_1(t),\gamma_2(t)$ ]{\includegraphics[width=0.333\textwidth]{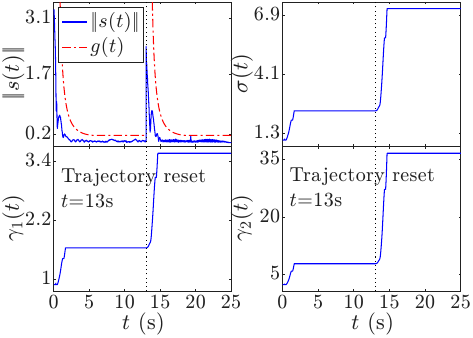}%
    \label{fig:gassta_fr3_reset_sliding_surface}}
    \subfloat[Case: \ref{case:initial_conditions}: Variation of initial conditions]
    {\includegraphics[width=0.333\textwidth]{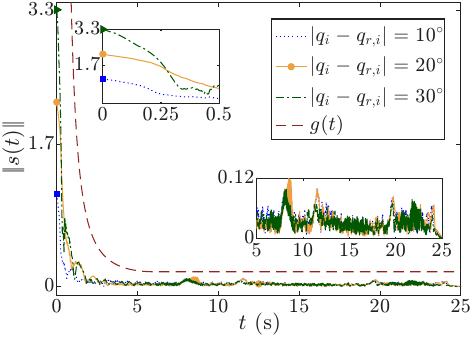}%
    \label{fig:gassta_fr3_exp_variation_initial_conditions}}
    \caption{Experimental results for the FR3 arm tracking a reference trajectory for cases \ref{case:tracking}-\ref{case:initial_conditions}, with $\alpha=0.7$ and $t_c=6$s. The dashed vertical line for case~\ref{case:trajectory_reset} results indicates the instant of reference trajectory reset at $t=13$ s.}
    \label{fig:gassta_fr3_exp_tracking_casea_b_c}
\end{figure*}

\begin{figure*}
    \centering
    \subfloat[Case: \ref{case:exp_prescribed_time}: Variation of $t_c=3,4,5,6$s]{\includegraphics[width=0.333\textwidth]{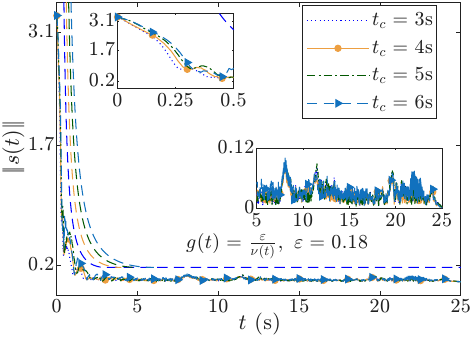}%
    \label{fig:gassta_fr3_exp_variation_prescribed_time}}
    \hfill
    \subfloat[Case: \ref{case:alpha}: control variation and $\norm{\bs{s}}$  RMSE]{\includegraphics[width=0.333\textwidth]{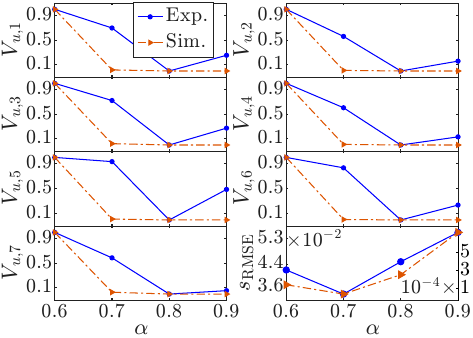}%
    \label{fig:gassta_fr3_exp_variation_alpha_input}}
    \hfill
    \subfloat[Case: \ref{case:payload_conditions}: payload variation]
    {\includegraphics[width=0.333\textwidth]{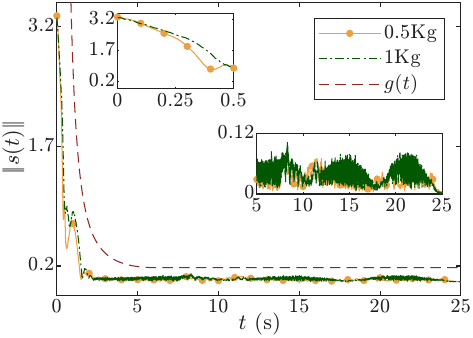}%
    \label{fig:gassta_fr3_exp_variation_payloads}}
    \hfill
    \subfloat[Case: \ref{case:suspended_load}: suspended load variation]{\includegraphics[width=0.333\textwidth]{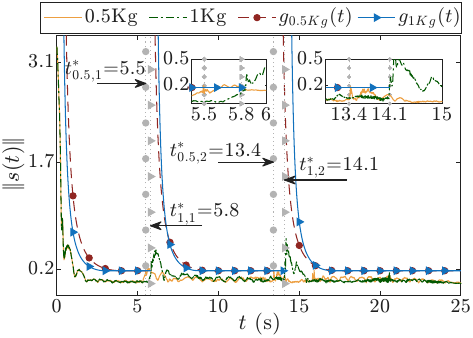}%
    \label{fig:gassta_fr3_exp_variation_suspended_load}}
    \hfill
    \subfloat[Case: \ref{case:suspended_load}: $\sigma(t),\gamma_1(t),\gamma_2(t)$]{\includegraphics[width=0.333\textwidth]{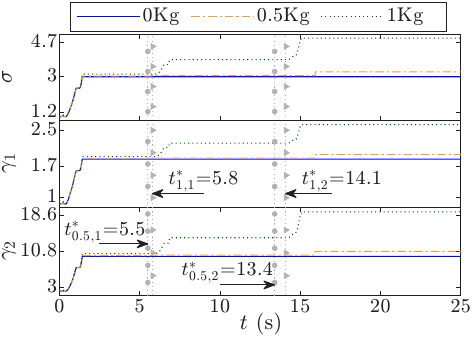}%
    \label{fig:gassta_fr3_exp_variation_suspended_alpha}}
    \subfloat[Case: \ref{case:suspended_load}: Control input in Nm]{\includegraphics[width=0.333\textwidth]{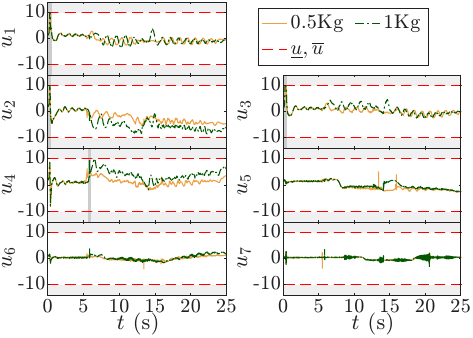}%
    \label{fig:gassta_fr3_exp_variation_suspended_load_input}}
    \hfill
    \caption{Experimental and simulation results for the FR3 arm tracking a minimum jerk reference trajectory for cases (3)--(4). Total Variation ($V_{u,i}$) is normalized per input channel for both simulation (Sim.) and experimental (Exp.) studies in Case: \ref{case:alpha}. For Case: \ref{case:suspended_load}, $t_{m,1}^*$ and $t^*_{m,2}$ denote the barrier reset instants following the taut-cable event and the manual perturbation respectively, so that $t_1^k=t_{m,n}^*$, with $m=\{0.5,\,1\}$ the load in Kg and $n=\{1,\,2\}$ the perturbation event. Input saturation occurs at or near these instants, with the dead-zone~\eqref{dead_zone} arresting gain adaptation over the saturated interval (Fig.~\ref{fig:gassta_fr3_exp_variation_suspended_alpha}); grey bands in Fig.~(f) indicate saturation.}
\end{figure*}

\emph{Nominal tracking and reference reset:} For Case~\ref{case:tracking},
Fig.~\ref{fig:gassta_fr3_sliding_surface} shows convergence of the sliding variable at
$t_c=6$\,s, with steady-state errors on the order of $1^\circ$
(Table~\ref{tab:franka_exp_metrics}) and the corresponding tracking responses in
Figs.~\ref{fig:gassta_fr3_jointPosition}--\ref{fig:gassta_fr3_jointVelocity}. This is
achieved under actuator saturation, as seen in Fig.~\ref{fig:gassta_fr3_input}. The
adaptive gain increases during the transient phase to compensate for the unknown uncertainty and
settles to a constant value in steady state, in accordance with \emph{Remark}~5. For
Case~\ref{case:trajectory_reset}, the reference reset at $t=13$\,s induces an
instantaneous tracking error from which the sliding variable re-converges
(Fig.~\ref{fig:gassta_fr3_reset_sliding_surface}), with steady-state accuracy preserved
(Table~\ref{tab:franka_exp_metrics},
Figs.~\ref{fig:gassta_fr3_reset_jointPosition}--\ref{fig:gassta_fr3_reset_jointVelocity}).

\emph{Independence from initial conditions and $t_c$:} Case~\ref{case:initial_conditions}
establishes that convergence to $\varepsilon$ occurs at the prescribed instant across
initial offsets of $10^\circ$--$30^\circ$
(Fig.~\ref{fig:gassta_fr3_exp_variation_initial_conditions}), while
Case~\ref{case:exp_prescribed_time} establishes the same across
$t_c=\{3,4,5,6\}$\,s (Fig.~\ref{fig:gassta_fr3_exp_variation_prescribed_time}). Together
these confirm the decoupling from initial conditions and system parameters asserted in
\emph{Remark}~7, delivered by the time-shifting mechanism~\eqref{modulating_gain_2} and the
adaptive gain~\eqref{modulating_gain} without knowledge of the
perturbation bounds.

\emph{Chattering reduction with $\alpha$:} Case~\ref{case:alpha} quantifies control chatter
through the total variation per input channel,
$V_{u,i} = \sum^N_{k=1}|u(t_{k+1}) - u(t_k)|/N$. Experimentally, $V_{u,i}$ decreases
monotonically over $\alpha\in[0.6,0.8]$ and rises slightly at $\alpha=0.9$
(Fig.~\ref{fig:gassta_fr3_exp_variation_alpha_input}), while remaining well below the
values at $\alpha=0.6$ and $0.7$; the rise is attributed to the larger adaptive
gain~$\gamma_1(t)$ amplifying encoder noise through $\gamma_1\lfloor\bs{s}\rceil^\alpha$
in~\eqref{control_policy}. A numerical study under identical conditions but without
measurement noise recovers monotonicity across the full range
(Fig.~\ref{fig:gassta_fr3_exp_variation_alpha_input}), confirming that the theoretical
trend holds and isolating noise as the source of the experimental deviation. Tracking
performance is essentially unaffected by $\alpha$, the RMSE of $\norm{\bs{s}}$ varying only
on the order of $10^{-2}$.

\emph{Robustness to smooth and non-smooth loads:} Case~\ref{case:payload_conditions}
applies payloads of $(0,0.5,1)$\,Kg from $t=0$
(Fig.~\ref{fig:fr3_experimental_setup}). The error metrics remain comparable to the
no-payload case (Table~\ref{tab:franka_exp_metrics},
Fig.~\ref{fig:gassta_fr3_exp_variation_payloads}), and no envelope violation occurs at
either mass: a payload applied smoothly is rejected by the compensation term $\rho_0(\bs{x})$
and the adaptive gain $\sigma(t)$ without recourse to the barrier mechanism.
Case~\ref{case:suspended_load} applies the same masses non-smoothly, and the behaviour
differs qualitatively. The taut-cable event transfers the load impulsively: the velocity
change occurs over a small number of control periods, producing a near-discontinuous
increase in $\norm{\bs{s}}$ that no bound on the disturbance derivative can accommodate.
The prescribed envelope is violated and a barrier reset is triggered at $t^*_{m,1}$
(Fig.~\ref{fig:gassta_fr3_exp_variation_suspended_load}), with the $1$\,Kg load producing a
substantially larger transient than the $0.5$\,Kg load. This is precisely the disturbance
class for which escape is terminal in the barrier designs
of~\cite{Obeid:2020,SPSTC}. Here the time-shifting
mechanism~\eqref{modulating_gain_2} reinitializes $g(t)$ at $t^1_1=t^*_{m,1}$, and
re-convergence to $\varepsilon$ follows by $t^*_c=t^1_1+t_c$ per \emph{Theorem}~3.1~(case
(b)); the manual perturbation at $t^*_{m,2}$ reproduces the same sequence
(Fig.~\ref{fig:gassta_fr3_exp_variation_suspended_alpha}). Both resets occur with multiple
input channels saturated
(Fig.~\ref{fig:gassta_fr3_exp_variation_suspended_load_input}), and the steady-state
metrics in Table~\ref{tab:franka_exp_metrics} are recovered in each instance without prior
knowledge of the perturbation bounds or their derivatives.

\subsection{Quantitative Comparisons}
This subsection compares the proposed scheme against the leading alternative
designs~\cite{SPSTC}, \cite{bertingo_7dof}, and \cite{Moreno:2022} on the FR3 arm; the
remaining schemes in Table~\ref{tab:qual_comparison} are excluded by the assumptions
recorded there. The comparison serves two purposes. The schemes in~\cite{SPSTC}
and~\cite{Moreno:2022} are included to demonstrate the consequences of the assumption
violations identified in Section~\ref{introduction}: \cite{SPSTC} is formulated for scalar
systems and does not account for multivariable coupling or state-dependent perturbations,
while \cite{Moreno:2022} requires the uncertainty to possess a bounded time derivative and
a particular matched structure. Neither hypothesis holds for the EL system considered
here. The scheme in~\cite{bertingo_7dof} is multivariable, achieves prescribed-time
tracking, and remains well-posed on this platform throughout; quantitative tracking
comparisons are accordingly made against~\cite{bertingo_7dof}.

\begin{figure*}
    \centering
    \subfloat[Simulation results for case \ref{case:tracking}]{\includegraphics[width=0.325\textwidth]{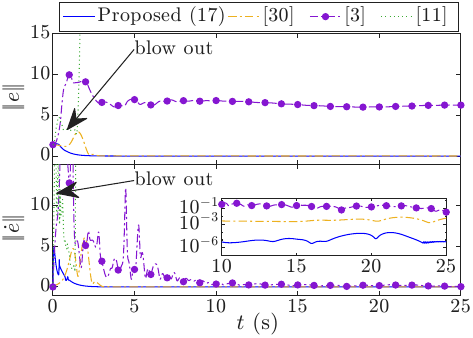}%
    \label{fig:gassta_fr3_sim_noDisturbancesStates}}
    \hfill
    \subfloat[Simulation results for case \ref{case:payload_conditions}]{\includegraphics[width=0.325\textwidth]{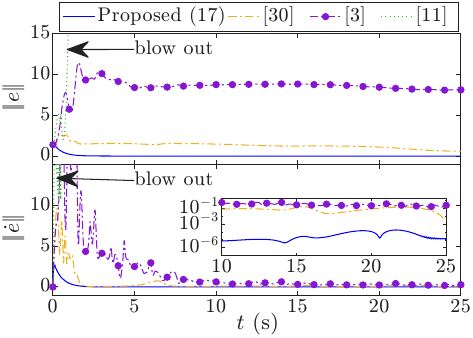}%
    \label{fig:gassta_fr3_sim_DisturbancesStates}}
    \hfill
    \subfloat[Experimental results for case \ref{case:tracking}]{\includegraphics[width=0.325\textwidth]{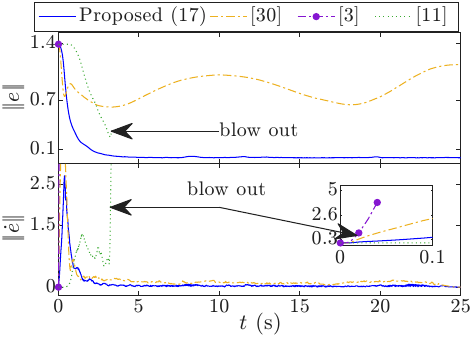}
    \label{fig:gassta_fr3_exp_comparisons}}
    \caption{Comparison of FR3 tracking using the proposed control policy~\eqref{control_policy} and the studies \cite{bertingo_7dof}, \cite{Moreno:2022}, and \cite{SPSTC}.}
    \label{fig:gassta_fr3_sim}
\end{figure*}
\subsubsection{Numerical Comparisons} \label{sec:fr3_gains} The FR3 arm with initial condition $\bs{q}(0) = [0, -\pi/4, 0, -3\pi/4, 0, \pi/2, \pi/4]^\top$ and $1$\,ms sampling period is considered, tracking a minimum jerk trajectory~\cite{Kyriakopoulos_Saridis_1994} offset by $30^\circ$ from the initial configuration. For the proposed controller, $M_0 = 2\bs{I}_7$, $\alpha=0.7$, $\gamma_{2}(0)=0.5$, $\sigma(0)=4$, $\eta_1=16$, $\eta_2=10^{-4}$, $\varepsilon=0.005$, $t_c=4$\,s, $\Gamma = 2\bs{I}_7$, $h=0.002$. For \cite{SPSTC}, $\bs{c}_1=2\bs{I}_7$, $\bs{c}_2=\bs{I}_7$, $h_1=h_2=4$, $\eta=5$, $\epsilon=0.18$. For \cite{bertingo_7dof}, $k+\theta = 0.2$, $\alpha=4$, $\zeta=0.4$, $\eta=5\times10^{-7}$, $T=2$\,s. For \cite{Moreno:2022}, $\alpha=\beta=0.01$, $b=2$, $p=0.3$, $k_1=8$, $k_2=4$.

Table~\ref{tab:franka_comparison_metrics} reports steady-state errors on the order of
$10^{-3}$ for the proposed scheme~\eqref{control_policy} in both joint position and
velocity, against $10^{-1}$ for~\cite{bertingo_7dof} and $10^{2}$
for~\cite{Moreno:2022}, corroborated by the tracking responses in
Fig.~\ref{fig:gassta_fr3_sim}. The margin over~\cite{bertingo_7dof} follows from the
absence of explicit compensation for unknown state-dependent and external perturbations in
that design, whereas the proposed adaptive gain compensates during the transient and
settles in steady state, avoiding conservative control effort. Under a $1$\,kg payload
attached to the end-effector, \cite{SPSTC} becomes unstable and \cite{Moreno:2022} tracks
poorly, both as predicted by the assumption violations noted above, while
\cite{bertingo_7dof} remains stable with noticeably degraded accuracy
(Fig.~\ref{fig:gassta_fr3_sim_DisturbancesStates}). The proposed
controller~\eqref{control_policy} maintains its nominal tracking performance, requiring
only that the lumped uncertainty be bounded in terms of known states and parametric
bounds.

\begin{table}[htbp]
     \centering
     \caption{Comparison of the MSE and RMSE for various control schemes considered in this study for the FR3 arm.}  
 \begin{tabular}{|c|c|c|c|c|c|}
     \hline
     \multirow{2}{*}{Platform}&\multirow{2}{*}{Method} & \multicolumn{2}{c|}{$\norm{\bs{e}}$ (deg)} & \multicolumn{2}{c|}{$\norm{\dot{\bs{e}}}$ (deg/s)} \\\cline{3-6}
      && MSE & RMSE & MSE & RMSE \\
     \hline
     \multirow{3}{*}{Sim. (Case \ref{case:tracking})}&\eqref{control_policy} & \textbf{0.008} & \textbf{0.003} & \textbf{0.005} & \textbf{0.002} \\\cline{2-6}
     &\cite{bertingo_7dof}  & 0.904 & 0.448 & 0.468 & 0.221 \\\cline{2-6}
     &\cite{Moreno:2022} & 393.068 & 363.829 & 161.496 & 34.683 \\\hline
    \multirow{3}{*}{Sim. (Case \ref{case:payload_conditions})}&\eqref{control_policy} & \textbf{0.019} & \textbf{0.007} & \textbf{0.011} & \textbf{0.005} \\\cline{2-6}
     &\cite{bertingo_7dof}  & 88.093 & 69.667 & 42.864 & 9.795 \\\cline{2-6}
     &\cite{Moreno:2022} & 505.040 & 489.302 & 168.603 & 35.344 \\\hline
     \multirow{2}{*}{Exp. (Case \ref{case:tracking})}&\eqref{control_policy} & \textbf{1.1} & \textbf{0.3} & \textbf{4.5} & \textbf{1.9} \\\cline{2-6}
     &\cite{bertingo_7dof}  & 65.1 & 51.7 & 10.2 & 6.4 \\\hline
     \multicolumn{6}{@{}l@{}}{\makecell[l]{%
Note: The control scheme in \cite{SPSTC} exhibits unstable performance and is\\ therefore not displayed for the simulations (Sim.).
Similarly, the controllers\\ in \cite{Moreno:2022} and \cite{SPSTC} are omitted for the experimental results (Exp.).
}}
 \end{tabular}
      \label{tab:franka_comparison_metrics}
 \end{table}

\subsubsection{Experimental comparison} The initial conditions and reference trajectory of Section~\ref{sec:fr3_gains} are used, at a $1$\,KHz control loop rate. For the proposed scheme, $M_0 = 2\bs{I}_7$, $\alpha=0.7$, $\gamma_{2}(0)=2.1$, $\sigma(0)=1$, $\eta_1=4.5$, $\eta_2=0.1$, $\varepsilon=0.18$, $t_c=6$\,s, $\Gamma = \text{diag}(2,3,2,3,3,2,1)$, $h=0.002$\,s. For \cite{SPSTC}, $\bs{c}_1=\text{diag}(2,3,2,3,3,2,1)$, $\bs{c}_2=\bs{I}_7$, $h_1=h_2=1.5$, $\eta=4$, $\epsilon=0.18$, $T=6$. For \cite{bertingo_7dof}, $k+\theta = 6$, $\alpha=8$, $\zeta=0.4$, $\eta=0.005$, $T=6$\,s. For \cite{Moreno:2022}, $\alpha=\beta=0.01$, $b=1$, $p=0.3$, $k_1=8$, $k_2=2$. The gains for \cite{SPSTC}, \cite{bertingo_7dof} and \cite{Moreno:2022} were tuned on the experimental platform to obtain the best achievable performance prior to the reported trials.

Despite this tuning, \cite{Moreno:2022} and \cite{SPSTC} diverged under nominal conditions
and were terminated on detecting the onset of instability
(Fig.~\ref{fig:gassta_fr3_exp_comparisons}), precluding a meaningful comparison under
payload perturbation; those comparisons are therefore reported in simulation above, where
all schemes remain well-posed. Against \cite{bertingo_7dof},
Table~\ref{tab:franka_comparison_metrics} shows the proposed
scheme~\eqref{control_policy} achieving substantially lower steady-state error. The margin
is wider than in the numerical study, as the experiments additionally expose joint encoder
measurement noise, to which \cite{bertingo_7dof} is comparatively more sensitive.
\section{Conclusion}
\label{section5}
This article presents a derivative-free generalized multivariable super-twisting framework for constrained Euler--Lagrange systems. Replacing the discontinuous sign-function integral term with a continuous fractional-power term removes the structural dependence on
the disturbance derivative, so that the admissible disturbance class extends to non-smooth signals that lie outside the reach of existing super-twisting designs. The state-dependent scaling of the perturbation is addressed jointly through gain adaptation and direct feedforward compensation, which together preserve stability without knowledge
of the perturbation magnitude. Input saturation is accommodated by a nonlinear equality transformation with a filtered saturation coefficient, and a time-shifting barrier mechanism restores the convergence guarantee whenever the trajectory escapes the
prescribed envelope. The resulting scheme converges in a user-prescribed time to a user-specified bound, independently of initial conditions and unknown system parameters, and detailed simulation and experimental studies on the FR3 manipulator confirm these properties under actuator saturation and non-smooth disturbances.


\end{document}